\documentclass[1p,review]{elsarticle}

\usepackage{algorithm}
\usepackage{algorithmic}
\usepackage{amsmath}
\usepackage{amsthm}
\usepackage{amssymb}
\usepackage{amsfonts}
\usepackage{float}
\usepackage{epsfig}
\usepackage{verbatim}
\usepackage{graphicx,epstopdf}
\usepackage{subfig}
\usepackage{nicefrac}
\usepackage{paralist}
\usepackage{color}
\usepackage{url}
\usepackage[greek,english]{babel}
\usepackage{eurosym} 

\theoremstyle{plain}

\newtheorem{theorem}{Theorem}

\newtheorem{definition}{Definition}
\newtheorem{corollary}{Corollary}
\theoremstyle{definition}

\floatname{algorithm}{Algorithm}

\allowdisplaybreaks[4]

\newcommand{\bkref}[1] {(\ref{#1})}

\journal{Elsevier Computer Communications}

\begin{document}

\begin{frontmatter}
\vspace*{-70px}
\title{A Distributed Demand-Side Management Framework for the Smart Grid}

\author[polimi]{Antimo Barbato}
\ead{antimo.barbato@polimi.it}
\address[polimi]{DEIB, Politecnico di Milano, via Ponzio 34/5, 20133 Milano, Italy.}

\author[polimi]{Antonio Capone\corref{cor1}}
\ead{antonio.capone@polimi.it}
\cortext[cor1]{Corresponding author, Tel: (+39) 02.2399.3449, Fax: (+39) 02.2399.3413}

\author[lri]{Lin Chen}
\ead{lin.chen@lri.fr}
\address[lri]{LRI, Université Paris-Sud, Bat. 650, rue Noetzlin, 91405 Orsay, France.}


\author[lri,iuf]{Fabio Martignon}
\ead{fabio.martignon@lri.fr}
\address[iuf]{IUF, Institut Universitaire de France}

\author[lipade]{Stefano Paris}
\ead{stefano.paris@parisdescartes.fr}
\address[lipade]{Paris Descartes University, 45 rue des Saints Peres, 75270 Paris, France.}


\begin{abstract}

This paper proposes a fully distributed Demand-Side Management system for Smart Grid infrastructures, especially tailored to reduce the peak demand of residential users.
In particular, we use a \textit{dynamic pricing strategy}, where energy tariffs are function of the overall power demand of customers.  
We consider two practical cases: (1) a fully distributed approach, where each appliance decides \textit{autonomously} its own scheduling, and (2) a hybrid approach, where each user must schedule all his appliances. We analyze numerically these two approaches, showing that they are characterized practically by the same performance level in all the considered grid scenarios.

We model the proposed system using a non-cooperative game theoretical approach, and demonstrate that our game is a generalized ordinal potential one under general conditions.
Furthermore, we propose a simple yet effective best response strategy that is proved to converge in a few steps to a pure Nash Equilibrium, thus demonstrating the robustness of the power scheduling plan obtained without any central coordination of the operator or the customers.
Numerical results, obtained using real load profiles and appliance models, show that the system-wide peak absorption achieved in a completely distributed fashion can be reduced up to 55\%, thus decreasing the capital expenditure (CAPEX) necessary to meet the growing energy demand.


\end{abstract}


\end{frontmatter}


\section{Introduction}
\label{introduction}


The electricity generation, distribution and consumption are in the throes of change due to significant regulatory, societal and environmental developments, as well as technological progress. Recent years have witnessed the redefinition of the power grid in order to tackle the new challenges that have emerged in electric systems. 
One of the most relevant challenges associated with the current power grid is represented by the \textit{peaks} in the power demand due to the high correlation among energy demands of customers. 
Since electricity grids have little capacity to store energy, power demand and supply must balance at all times; as a consequence, energy plants capacity has to be sized to match the total demand peaks, driving a major increase of the infrastructure cost, which remains underutilized during off-peak hours. 
This waste of resources has become an even more critical issue in the last few years due to the increase of the worldwide energy consumption \cite{european2010europe} and the increasing share of renewable energy sources \cite{wilkes2011wind}. 
High energy peaks are mostly due to residential users, who cover a relevant portion of the worldwide energy demands \cite{EEU}, but are inelastic with respect to the grid requirements as they usually run their home appliances only depending on their own requirements. For this reason, residential users can play a key role in addressing the peak demand problem. Time-Of-Use (TOU) tariffs represent a clear attempt
to incite users to shift their energy loads out of the peak hours \cite{tsg_2013_tou}.

The most promising solution to tackle the peak demand challenge is represented by the Smart Grid, in which an intelligent infrastructure based on Information and Communication Technology (ICT) tools is deployed alongside with the distribution network, which can deal with all the decision variables while minimizing the effort required to end-users.
All data provided by the grid, such as the consumption of buildings \cite{jiang2009design} \cite{bressan2010deployment}, electricity costs and distributed Renewable Energy Sources (RESs) data, can be used to optimize its efficiency through \textit{Demand-Side Management} (DSM) methods, which represent a proactive approach to manage the household electric devices by integrating customers' needs and requirements with the retailers' goals \cite{gellings1987demand}. The main objective  of these methods is to modify consumers' energy demand in a proper way by deciding \textit{when} and \textit{how} to execute home appliances so as to improve the overall system efficiency while guaranteeing low costs and high comfort to users.

In this paper we propose a novel, \textit{fully distributed DSM system} aimed at reducing the peak demand of a group of residential users (e.g., a smart city neighborhood).
In particular, we consider a \textit{dynamic pricing strategy}, where energy tariffs are function of the overall power demand of customers. 

We model our system using a game theoretical approach, considering two practical cases where (1) each appliance decides \textit{autonomously} its scheduling in a fully distributed fashion (Single-Appliance DSM), and (2) each user must schedule all his home appliances (Multiple-Appliance DSM).
The proposed approach automatically ensures the reduction of the electricity demand at peak hours due to dynamic pricing.

We compare numerically these two cases, showing that the first is characterized only by a negligible performance degradation in all the considered grid scenarios.
Nevertheless, while both mechanisms achieve almost the same performance level, the Multiple-Appliance DSM system requires a more complex architecture with a central server for each house that collects all appliances information and plays on behalf of the householder.
Such an approach would increase the installation and operating costs due to the higher system complexity.
On the contrary, in the Single-Appliance DSM system, one can use the processing and communication capabilities of devices that can autonomously optimize their usage, thus greatly simplifying the architecture design and system configuration. This solution is made possible by the diffusion of \textit{Smart Appliances} that are no longer merely 
passive devices, but active participants in the power grid infrastructure \cite{lui2010get}. 

We underline that, while recent literature has focused on the design of DSM systems for \textit{controllable} devices~\cite{mohsenian2010autonomous}, namely devices whose power load profile within their operating time can be modulated according to the DSM goals, our work designs a distributed DSM to select the best (cheapest) schedule for \textit{shiftable} appliances.
Indeed, differently from air conditioning or heating systems, appliances like the washing machine and the electric oven have a fixed power profile optimized for specific goals.
In such cases, a user can choose the starting time for each shiftable appliance, whose power profile is \textit{fixed}.
Therefore, our scheme is complementary to the approaches devised for controllable devices.

We demonstrate that our game is a \textit{generalized ordinal potential game} \cite{MondererShapley} under some simple and very general conditions (viz., the regularity of the pricing function). Such feature guarantees some nice properties, such as the existence of at least one pure Nash equilibrium (where no player has an incentive to deviate unilaterally from the scheduling pattern he decided upon).
Furthermore, we show that any sequence of asynchronous improvement steps is finite and always converges to a pure Nash equilibrium.


In summary, our paper makes the following contributions:

\begin{itemize}
\item The proposition of a novel, fully distributed DSM method, able to reduce the peak demand of a group of residential users, which we model and study using a game theoretical framework. In our vision, the energy retailer fixes the energy price dynamically, based on the total power demand of customers; then, appliances autonomously decide their schedule, reaching an efficient Nash equilibrium point. 

\item Mathematical proofs that our proposed game is a \textit{generalized ordinal potential game}, under general conditions.

\item The demonstration of the Finite Improvement Property, according to which any sequence of asynchronous improvement steps (and, in particular, \textit{best response dynamics}) converges to a pure Nash equilibrium.



\item A thorough numerical evaluation that shows the effectiveness of the proposed approach in several scenarios, with real electric appliances scheduled by householders.

\end{itemize}

The paper is organized as follows. Section~\ref{related_work} discusses related work. Section~\ref{system_model} describes the main characteristics of the distributed system we propose to manage the energy consumption of residential users. Section~\ref{distributed_ps} presents our proposed game theoretical formulation for the Single and Multiple-Appliance DSM, as well as the structural properties of our game.
Numerical results are presented and analyzed in Section~\ref{Numerical_Results}.
Finally, Section~\ref{conclusion} concludes the paper.


\section{Related Work}
\label{related_work}

Demand-Side Management (DSM) mechanisms have recently gained attention by the scientific community due to their advantages in terms of wise use of energy and cost reduction~\cite{strbac2008demand}.
In DSM systems proposed in the literature, a mechanism is defined that, based on energy tariffs and data forecasts for future periods (e.g., photovoltaic power generation, devices future usage), is able to automatically and optimally schedule the home devices activities for future periods and to define the whole energy plan of users (i.e., when to buy and sell energy to the grid). The main goal of these solutions is to minimize the electricity costs while guaranteeing the users' comfort; this can be achieved through the execution of methods based on optimization models \cite{JacominoLe4OR}, \cite{Agnetis2011} or heuristics, such as Genetic Algorithms \cite{soares2013domestic} and customized Evolutionary Algorithms \cite{allerding2012electrical}, which are used to solve more complex formulations of the demand management problem. Since RESs diffusion is rapidly increasing, several works include renewable plants into DSM frameworks. In these cases, devices are scheduled also based on the availability of an intermittent electricity source (e.g., PV plants) and users' profits from selling renewable electricity to the energy market are taken into account \cite{Clasters2010}. The uncertainty of RESs generation forecasts is tackled through stochastic approaches, such as stochastic dynamic programming which is a very suitable tool to address the decision-making process of energy management systems in presence of uncertainty, such as the one related to the electricity produced from weather-dependent generation sources \cite{livengood2009energy}. The efficiency of demand management solutions can be notably improved by including storage systems that can increase the DSM flexibility in optimizing the usage of electric resources. Specifically, batteries can be used to harvest the renewable generation in excess for later use or to charge the ESS when the electricity price is low, with the goal of minimizing the users' electricity bill \cite{guo2012optimal}. \\
Solutions \cite{JacominoLe4OR}--\cite{guo2012optimal} are based on a \textit{single-user} approach in which the energy plans of residential customers are individually and locally optimized.
However, in order to achieve relevant results from a system-wide perspective, the energy management problem could be applied to \textit{groups} of users (e.g., a neighborhood or micro-grids), instead of single users.
For this reason, some preliminary solutions have been proposed in the literature to manage energy resources of groups of customers.
In \cite{barbato2011house}, for example, the energy bill minimization problem is applied to a group of cooperative residential users equipped with PV panels and storage devices (i.e., electric vehicle batteries).
A global scale optimization method is also proposed in \cite{molderink2009domestic}, in which an algorithm is defined to control domestic electricity and heat demand, as well as the generation and storage of heat and electricity of a group of houses.
These multi-user solutions require some sort of centralized coordination system run by the operator in order to collect all energy requests and find the optimal solution.
To this end, a large flow of data must be transmitted through the Smart Grid network, thus introducing scalability constraints and requiring the definition of high-performance communication protocols. 
Furthermore, the coordination system should also verify that all customers comply with the optimal task schedule, since the operator has no guarantee that any user can gain by deviating unilaterally from the optimal solution.
Therefore, the collection of users' metering data and the enforcing of the optimal appliance schedule can introduce novel threats to customers' security and privacy.
For these reasons, some \textit{distributed} DSM methods have been proposed in which decisions are taken locally, directly by the end consumer.
In this case, Game Theory represents the ideal framework to design DSM solutions.
Specifically, in~\cite{mohsenian2010autonomous} a distributed DSM system among users is proposed, where the users' energy consumption scheduling problem is formulated as a game: the players are the users, and their strategies are the daily schedules of their household appliances and loads.
The goal of the game is to either reduce the peak demand or the energy bill of users.
A game theoretical approach is also used in \cite{ibars2010}, in which a distributed load management is defined to control the power demand of users through dynamic pricing strategies.
However, in these works, a very simplified mathematical description is used to model houses, which does not correspond to real use cases. 

In this paper we propose a DSM method, based on a game theoretical approach, which overcomes the most important limitations of the works proposed in the literature and described above. Our DSM is a fully distributed system, in which no centralized coordination is required, and only a limited and aggregated amount of data needs to be transmitted between the operator and the householders through the Smart Grid. For these reasons, scalability, communication, privacy and security issues are greatly mitigated.
Moreover, a realistic model of household contexts is illustrated; specifically, a mathematical description  of home devices is provided. Devices are defined as non-preemptable activities characterized by specific load consumption profiles, determined based on real data, and are scheduled according to users' preferences defined based on real use-case scenarios. Finally, to the best of our knowledge, the single-appliance demand management game proposed in this paper, in which electric devices can autonomously and locally optimize their usage, has never been studied in the literature. 

\section{System Model}
\label{system_model}

The power scheduling system here proposed is designed to manage the electric appliances of a group  of residential users consisting of a set $\cal H$ of houses (e.g., a smart city neighborhood). This system is used to schedule the energy plan of the whole group of users over a 24-hour time horizon based on a \textit{fully distributed} approach, with the final goal of improving the efficiency of the whole power grid by reducing the peak demand of electricity, while still complying with users' needs and preferences. More specifically, in our model we represent the daily time as a set $\cal T$ of time slots. Each householder \footnote{In this paper, we use the terms \textit{householder} and \textit{user} interchangeably.} $h\in{\cal H}$ has a set of non-preemptive electric appliances, ${\cal A}$, that must be executed during the day. In particular, the load profile of each appliance is modeled as an ordered sequence of phases, $\cal F$, in which a certain amount of power is consumed. We assume that the power consumption $l_{ahf}$ of a device $a\in{\cal A}$ belonging to user $h\in{\cal H}$ in each phase $f \in {\cal F}$ is an average of the real consumption of the device within the time slot duration (see Figure~\ref{lppower}, where 15-minute phases are used for a washing machine ~\cite{micene}).

Each device $a$ of user $h$ needs to run for $d_{ah}$ \textit{consecutive} slots within a total of ${\cal R}_{ah}$ slots delimited by a minimum starting time slot, $ST_{ah}$, and a maximum ending time slot, $ET_{ah}$ (verifying the constraint $ST_{ah} \leq ET_{ah} - d_{ah} + 1 $). These two parameters, $ST_{ah}$ and $ET_{ah}$, represent the users' preferences in starting each home device; they can be directly provided by users or automatically obtained through learning algorithms such as the one presented in \cite{2013_greencom}.

In our model, we consider two different kinds of devices: 

\begin{itemize}
\item \textit{Shiftable} appliances (e.g., washing machine, dishwasher): they are manageable devices that must be scheduled and executed during the day. 
In particular, for each shiftable device $a\in{\cal }A$ of the householder $h\in{\cal H}$, the minimum starting time and the maximum ending time verify the constraint $ST_{ah} < ET_{ah} - d_{ah} + 1$.
Hence, their scheduling is an optimization variable in our model.

\item \textit{Fixed} appliances (e.g., light, TV) are non-manageable devices, for which the starting/ending times are fixed and cannot be optimized.
More specifically, for each fixed device $a\in{\cal A}$ of the householder $h\in{\cal H}$, the minimum starting time and the maximum ending time verify the constraint $ST_{ah} = ET_{ah} - d_{ah} + 1$.
\end{itemize}

\begin{figure}[t!]
\begin{center}
	\includegraphics[width=0.45\textwidth]{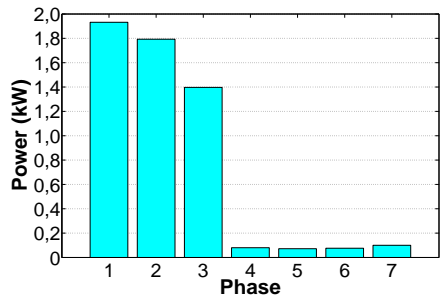}
	\caption{\small{Example of a load profile $l_{ahf}$ of a washing machine.}}
\label{lppower}
\end{center}
\end{figure}

Devices scheduling is represented by the binary variable $x_{aht}$, which is defined for each appliance $a \in {\cal A}$ of each householder $h\in{\cal H}$, and for each time slot $t \in \cal T$. It is equal to 1 if appliance~$a$ starts in time slot $t$, 0 otherwise.
In order to use home appliances, householders can buy energy from the electricity retailer.
In particular, the power demand of user~$h$ at time $t$ is denoted by $y_{ht}$. The power demand of each user cannot exceed a supply limit defined by the retailer and denoted by $\pi_{SL}$; this limit represents the maximum power that can be used at any time. 

In our model, we have decided to use a dynamic pricing approach to define the electricity tariff since it represents a very promising method to improve the efficiency of the whole power grid \cite{Geelen2013151}. Since the higher the demand of electricity, the larger the capacity of grid generation and distribution to install, we suppose that the price of electricity at time $t$, $c_t(\cdot)$, is an increasing function of the total demand, $y_t$, of the group of users $\cal H$ at time $t$. Specifically, if the power demand is lower than a threshold $\pi_{TT}$, $c_t(\cdot)$ is a strictly increasing function of~$y_t$, otherwise it becomes a constant function of value $c_t(\pi_{TT})$. 

The objective of the power scheduling system is to reduce the daily bill of each user, by optimally scheduling house appliance activities and managing the power absorption from the grid.
Intuitively, based on the definition of the electricity price, by reducing the users' bills, the model is able to decrease the corresponding peak demand. 

Table~\ref{tab:notation} summarizes the notation used in the paper.
%
%
\begin{table}[ht!]
\centering
\caption{\small{Basic notation used in the paper.}}
\label{tab:notation}
\small
\begin{tabular}{|l|l|}
\hline
	$x_{aht}$			&	Binary variable that indicates if appliance $a$\\
						&	of householder $h$ starts its execution at time $t$\\
\hline
$y_{ht}$   & Power demand of user~$h$ at time $t$ \\
\hline
	$y_{t}$ 				&	Total power demand at time $t$ \\
\hline
	$c_t(\cdot)$ 		& 	Pricing function\\
\hline
	$\pi_{TT}$ 			& 	Tariff Threshold of the pricing function\\
\hline
	$\pi_{SL}$	 		& 	Power Supply Limit\\
\hline
	$l_{ahf}$ 			& 	Consumption of device $a$ of user $h$ in phase $f$\\
\hline
	$d_{ah}$ 			& 	Operating time slots for device $a$ of user $h$\\
\hline
	$ST_{ah}/ET_{ah}$	& 	Starting/Ending time for device $a$ of user $h$\\
\hline
\end{tabular}
\end{table} 
%
%

\section{Distributed Power Scheduling as a Non-cooperative Game}
\label{distributed_ps}

%
%

In this section, we model the distributed power scheduling problem, which constitutes the core of our proposed DSM system, using a non-cooperative game theoretical approach (formally described in Definition~\ref{def:game}), which naturally captures the interactions in such a distributed decision making process.
Our design rationale (Subsection \ref{singleAppliance}) is the following: each appliance $a\in{\cal A}$ is an \textit{autonomous decision maker} (or player) that must select the starting time of its execution (i.e., the $x_{aht}$ value); this permits to minimize the coordination required by a central server that would operate at each house to aggregate all appliances load and scheduling constraints.
Consequently, each appliance $a$ decides autonomously when to buy energy from the grid (i.e., $y_{ht}$) in order to minimize its contribution to the overall bill charged to house $h\in{\cal H}$, according to his user's\footnote{In this paper users are house owners, therefore we use interchangeably the terms \emph{house} and \emph{user}.} needs.

Then, after having solved the single-appliance game and studied its structural properties (Subsection \ref{teoremi}), in Subsection \ref{multiAppliance} we consider a natural (and more complex) extension where a player represents an entire household which \textit{jointly} decides the schedule of all his appliances.


\subsection{Single-Appliance Game Formulation}
\label{singleAppliance}
%
%
We first start describing the scenario where each appliance $a\in{\cal A}$ of house $h \in {\cal H}$ is modeled as an autonomous player in the power scheduling game $G$, which is defined as a triple $\{{\cal N}, {\cal I}, {\cal P}\}$: ${\cal N}={\cal A}\times{\cal H}$ is the player set, ${\cal I}\triangleq \{{\cal I}_n\}_{n\in {\cal N}}$ is the strategy set with ${\cal I}_n\triangleq \{x_{nt}\}_{n \in {\cal N}}$ being the strategy of player $n$, ${\cal P}\triangleq \{P_n\}_{n\in {\cal N}}$ is the cost function of player~$n$ with $P_n$ being the total price paid by $n$ for its electricity consumption (the total price due to appliance $a \in {\cal A}$ of house $h \in {\cal H}$).
Each appliance (player) $n$ chooses its strategy ${\cal I}_n$ to minimize its cost $P_n$.

The feasible power scheduling alternatives that form the strategy space ${\cal I}_n$ of each player $n = (a,h)$ (i.e., each appliance $a$ of householder $h$) must satisfy both the consumer needs and energy supply limits.
Specifically, the strategy space~${\cal I}_n$ must satisfy the following set of constraints:
\small
\begin{align}
	 {\cal I}_n = \bigg\lbrace
		& \overrightarrow{x}_{n} = \left[ x_{n1} ... x_{nt} ... x_{n \vert {\cal T}\vert} \right] \in \{0,1\}^{\vert {\cal T}\vert} : \nonumber\\
		& 	\sum_{t=ST_{n}}^{ET_{n} - d_{n} + 1} x_{nt} = 1 \label{const_1}\\
		& 	y_{nt} = \sum_{f \in {\cal F}: f \leq t} l_{nf} x_{n(t-f+1)} \qquad \qquad \forall t \in {\cal T}\label{const_3}\\
		& 	y_{ht} = \sum_{a \in {\cal A}} \sum_{f \in {\cal F}: f \leq t} l_{ahf} x_{ah(t-f+1)} \qquad \forall t \in {\cal T} \label{const_4}\\
		& 	y_{ht} \leq \pi_{SL} \qquad \qquad \qquad \forall  t \in {\cal T} \label{const_5}
		 \bigg\rbrace.
\end{align}
%
%
%
%
\normalsize
Constraints~\bkref{const_1} guarantee that appliance $n$ starts in exactly one time slot and it is carried out in the interval $(ST_{n}, ET_{n})$.
Constraints~\bkref{const_3} determine the daily consumption profile of the appliances in each time slot, which depends on their scheduling. More specifically, the power required by each appliance in each time slot $t$, $y_{nt}$, is equal to the load profile $l_{nf}$  of the phase carried out at time $t$. Note that a phase $f$ is running in $t$, only if the appliance started at time $t-f+1$, thus if $x_{n(t-f+1)} = 1$. In a similar fashion, equations~\bkref{const_4} define the daily power demand of house $h$ based on the appliances scheduling. Finally, constraints~\bkref{const_5} limit the overall power consumption of each house, since in every time slot $t \in \cal T$, the electricity bought from the grid cannot exceed the Supply Limit (SL) defined by the retailer and denoted by $\pi_{SL}$.
In such constraints, the power required by each appliance in each time slot $t$, $y_{nt}$, is equal to the load profile $l_{nf}$  of the phase executed starting from the time slot where $x_{nt} = 1$. Note that~\bkref{const_3} is used by the appliance $a$ to compute and minimize its contribution to the overall price charged to house $h$, whereas~\bkref{const_4} is used by householder $h$ to compute the bill.

Having defined the strategy space of each player, we can now define the \textit{single-appliance power scheduling game}.

\begin{definition}[Power Scheduling Game]
Mathematically, the power scheduling game is formalized as follows:
\small
\begin{align}
	G:\ & \min_{{\cal I}_n} \ P_n({\cal I}_n, {\cal I}_{-n}) = \sum_{t \in {\cal T}} y_{nt} \cdot c_{t}(y_{t}), \ \forall n\in{\cal N}. 
\label{eq:app_game}
\end{align}
\label{def:game}
\end{definition}
\normalsize
The solution of the power scheduling game is characterized by a Nash Equilibrium (NE), a strategy profile ${\cal I}^*=({\cal I}_n^*, {\cal I}_{-n}^*)$ from which no player has an incentive to deviate unilaterally, i.e.,
$$P_n({\cal I}_n^*, {\cal I}_{-n}^*)\ge P_n({\cal I}_n, {\cal I}_{-n}^*), \quad \forall n\in{\cal N}, \forall {\cal I}_n \in {\cal I}.$$

To study the efficiency of the NE(s) of $G$, we define the \textit{social cost} of all players as the total price, $P$, paid by all customers to the electricity retailer, as a function of ${\cal I}=\{{\cal I}_n\}_{n\in {\cal N}}$, where the strategy of player $n$ is ${\cal I}_n=\{x_{nt}\}_{n \in {\cal N}}$:
\small
\begin{eqnarray}
\label{eq.SocialCost}
	P({\cal I})=\sum_{h\in{\cal H}} \sum_{t \in \cal T} y_{ht} \cdot c_{t}(y_{t}),
\end{eqnarray}
where $y_{ht}$ is a function of $x_{nt}, n=(a,h) \in {\cal A}\times{\cal H}$ and $c_t$ is a function of $y_t$ that represents the total power demand of all players at time $t$.
%
%
\normalsize

By analyzing the utility functions of $G$, we can see that the pricing function $c_t(y_t)$ plays an important role on the resulting system equilibrium point(s). Specifically, our objective is to devise smart pricing policies to drive the system equilibrium to the optimum in terms of social cost. In this regard, we focus on a class of pricing functions, termed as \textit{regular pricing functions}, defined as follows.

\begin{definition}[Regular Pricing Function]
The pricing function $\{c_t(y_t)\}_{0\le t\le T}$ is a regular pricing function if the following properties hold:
\begin{itemize}
\item $c_t(y_t)$ is continuous, non-decreasing for $0\le t\le T$ and its derivative $c_t'(y_t)$ is continuous in $y_t$;
\item Given any two time intervals $[t_u^0,t_u^1]$, $[t_v^0,t_v^1]$ and power demand in these intervals $\{y_u\}_{t_u^0<u<t_u^1}$, $\{y_v\}_{t_v^0<v<t_v^1}$, if $\sum_{u=t_u^0}^{t_u^1} c_u'(y_u) > \sum_{v=t_v^0}^{t_v^1} c_v'(y_v)$, then it holds that $\sum_{u=t_u^0}^{t_u^1} [y_u c_u(y_u)]' > \sum_{v=t_v^0}^{t_v^1} [y_v c_v(y_v)]'$.
\end{itemize}
\end{definition}

\noindent\textbf{Remark}: Regular pricing functions characterize a family of utility functions widely applied in practical applications. A typical example of regular pricing function is the power function $c_t=\alpha y_t^{\beta}$ where $\alpha > 0$ and $\beta\ge 1$. The design motivation hinging behind such pricing functions is to encourage users to balance their electricity demand and consequently decrease the peak demand. 

In the following analysis, we show that under the condition that the pricing policy can be expressed by a regular function, the power scheduling game $G$ admits a number of desirable properties, particularly from the perspective of social cost.

\subsection{Solving the Power Scheduling Game}
\label{teoremi}

In this subsection, we solve the power scheduling game $G$ and study the structural properties of the game. We are specifically interested in large systems where the impact of an individual user on the system dynamics is limited.
Theorem~1 shows that $G$ is a generalized ordinal potential game, whose definition is reported hereafter for completeness.

\begin{definition}[Generalized Ordinal Potential Game]
Given a finite strategic game $\Gamma\triangleq \{{\cal N}, \{S_n\}_{n\in{\cal N}}, \{u_n\}_{n\in{\cal N}}\}$, $\Gamma$ is a generalized ordinal potential game if there exists a function (called potential function) $\Phi: S\rightarrow \mathbb{R}$ such that for every player $n\in{\cal N}$ and every $s_{-n}\in S_{-n}$ and $s_n,s_n'\in S_n$, it holds that
$$u_n(s_n,s_{-n})>u_n(s_n',s_{-n}) \Longrightarrow \Phi(s_n,s_{-n})>\Phi(s_n',s_{-n}).$$
\end{definition}

\begin{theorem}
Under the condition that $\{c_t(y_t)\}_{0\le t\le T}$ is a regular pricing function, the power scheduling game $G$ is a generalized ordinal potential game with the corresponding potential function being $P({\cal I})$
\label{theorem:potential_game}
%
\end{theorem}

\begin{proof}
To prove the theorem, it suffices to show that for any two strategies ${\cal I}_n$ and ${\cal I}_n'$ and for any player $n\in{\cal N}$, it holds that
$$P_n({\cal I}_n,{\cal I}_{-n})>P_n({\cal I}_n',{\cal I}_{-n}) \Longrightarrow P({\cal I}_n,{\cal I}_{-n})>P({\cal I}_n',{\cal I}_{-n}).$$

In this regard, assume that $P_n({\cal I}_n,{\cal I}_{-n})>P_n({\cal I}_n',{\cal I}_{-n})$.
Assume that $n$ (i.e., appliance $a \in {\cal A}$ of house $h \in {\cal H}$) starts its activity in time interval $t_u^0<u<t_u^1$ ($t_v^0<v<t_v^1$, respectively) in strategy ${\cal I}_n$ (${\cal I}_n'$).
Let $y_t$ denote the total power demand at time $t$ under strategy profile $({\cal I}_n,{\cal I}_{-n})$.
Between the strategy profiles $({\cal I}_n,{\cal I}_{-n})$ and $({\cal I}_n',{\cal I}_{-n})$, the difference is that $n$ migrates its power demand of $p_n$ from time interval $[t_u^0,t_u^1]$ to $[t_v^0,t_v^1]$.
Since we are focused on large systems where the impact of an individual user on the system dynamics is limited, i.e., $p_n\ll y_t$, it holds that
\small
\begin{align}
	& P_n({\cal I}_n,{\cal I}_{-n}) - P_n({\cal I}_n',{\cal I}_{-n}) =
	\sum_{u=t_u^0}^{t_u^1} p_n c_u(y_u) + \sum_{v=t_v^0}^{t_v^1} p_n c_v(y_v) \nonumber\\
	& - \left[ \sum_{u=t_u^0}^{t_u^1} p_n c_u(y_u - p_n) + \sum_{v=t_v^0}^{t_v^1} p_n c_v(y_v + p_n) \right] \simeq \nonumber\\
    & \simeq p_n \left[\sum_{u=t_u^0}^{t_u^1} c_u'(y_u - p_n) - \sum_{v=t_v^0}^{t_v^1} c_v'(y_v)\right] > 0,
\label{eq:aux1}
\end{align}
\normalsize
following the assumption that $P_n({\cal I}_n,{\cal I}_{-n})>P_n({\cal I}_n',{\cal I}_{-n})$.

Recalling the definition of regular pricing functions, it then holds that
\small
\begin{align}
	&\sum_{u=t_u^0}^{t_u^1} [(y_u - p_n) c_u(y_u - p_n)]' > \sum_{v=t_v^0}^{t_v^1} [y_v c_v(y_v)]'.
\label{eq:aux2}
\end{align}
\normalsize

On the other hand, we study the social cost under the strategy profiles $({\cal I}_n,{\cal I}_{-n})$ and $({\cal I}_n',{\cal I}_{-n})$.
Specifically, we can derive the difference between $P({\cal I}_n,{\cal I}_{-n})$ and $P({\cal I}_n',{\cal I}_{-n})$ as follows:
\small
\begin{align}
	& P({\cal I}_n,{\cal I}_{-n}) - P({\cal I}_n',{\cal I}_{-n}) = \sum_{u=t_u^0}^{t_u^1} [y_u c_u(y_u)] + \sum_{v=t_v^0}^{t_v^1} [y_v c_v(y_v)] \nonumber\\
	& - \left\{ \sum_{u=t_u^0}^{t_u^1} [(y_u - p_n) c_u(y_u - p_n)] + \sum_{v=t_v^0}^{t_v^1} [(y_v + p_n) c_v(y_v + p_n)] \right\} \nonumber\\
	& = \sum_{u=t_u^0}^{t_u^1} [y_u c_u(y_u)]-\sum_{u=t_u^0}^{t_u^1} [(y_u - p_n) c_u(y_u - p_n)] \nonumber\\
	& - \left\{ \sum_{v=t_v^0}^{t_v^1} [(y_v + p_n) c_v(y_v + p_n)] - \sum_{v=t_v^0}^{t_v^1} [y_v c_v(y_v)] \right\}.
\label{eq:aux3}
\end{align}
\normalsize

With some algebraic operations, we have
\small
\begin{align}
\label{equazione13}
 	& \sum_{u=t_u^0}^{t_u^1} [y_u c_u(y_u)]-\sum_{u=t_u^0}^{t_u^1} [(y_u - p_n) c_u(y_u - p_n)] = \nonumber\\
 	& \sum_{u=t_u^0}^{t_u^1} \{y_u[c_u(y_u) - c_u(y_u - p_n)] + p_n c_u(y_u - p_n)\} \simeq \nonumber\\
	& \simeq \sum_{u=t_u^0}^{t_u^1} \{y_u p_n c_u'(y_u - p_n) + p_n c_u(y_u - p_n)\} > \nonumber\\
	& > \sum_{u=t_u^0}^{t_u^1} [(y_u - p_n) c_u (y_u - p_n)]'.	
\end{align}
\normalsize

Similarly, we have
\small
\begin{align}
\label{equazione14}
	& \sum_{v=t_v^0}^{t_v^1} [(y_v + p_n) c_v (y_v + p_n)] - \sum_{v=t_v^0}^{t_v^1} [y_v c_v(y_v)] < \sum_{v=t_v^0}^{t_v^1} [y_v c_v (y_v)]'.
\end{align}
\normalsize

Hence, it follows from~\eqref{equazione13} and~\eqref{equazione14} that
\small
\begin{align}
	& P({\cal I}_n,{\cal I}_{-n}) - P({\cal I}_n',{\cal I}_{-n}) = \nonumber\\
	& \sum_{u=t_u^0}^{t_u^1} [y_u c_u(y_u)] - \sum_{u=t_u^0}^{t_u^1} [(y_u - p_n) c_u(y_u - p_n)] \nonumber\\
    & - \left\{ \sum_{v=t_v^0}^{t_v^1} [(y_v + p_n) c_v(y_v + p_n)] - \sum_{v=t_v^0}^{t_v^1} [y_v c_v(y_v)] \right\} > \nonumber\\
    & > \sum_{u=t_u^0}^{t_u^1} [(y_u - p_n) c_u (y_u - p_n)]' - \sum_{v=t_v^0}^{t_v^1} [y_v c_v(y_v)]'> 0.
\label{eq:aux4}
\end{align}
\normalsize

The proof is thus completed.
\end{proof}

\begin{corollary}[Efficiency of the Equilibrium]
Under the conditions of Theorem~\ref{theorem:potential_game}, the equilibrium of $G$ minimizes the total price paid to the operator, i.e., the total social cost.
\end{corollary}

\begin{corollary}[Convergence to the Equilibrium]
Under the conditions of Theorem~\ref{theorem:potential_game}, $G$ admits the Finite Improvement Property (FIP).
Any sequence of asynchronous improvement steps is finite and converges to a pure equilibrium. Particularly, the sequence of best response updates converges to a pure equilibrium.
\end{corollary}

Potential games have nice properties, such as existence of at least one pure Nash equilibrium, namely the strategy that minimizes $P(\cal I)$.
Furthermore, in such games, best response dynamics always converges to a Nash equilibrium.

Hereafter, we describe a simple implementation of best response dynamics, which allows each player $n$, namely each appliance $a$ of each householder $h$, to improve its cost function in the proposed power scheduling game.
Such algorithm is the best response strategy for a player $n$ minimizing objective function~\bkref{eq:app_game}, $\sum_{t \in \cal T} y_{nt} \cdot c_{t}(y_{t})$, assuming other appliances are not changing their strategies.

Specifically, each appliance, in an iterative fashion, defines its optimal power scheduling strategy based on electricity tariffs (calculated according to other players' strategies) and broadcasts its energy plan (i.e., its daily power demand profile) to the group ${\cal N}$.
At every iteration, energy prices are updated according to the last strategy profile and, as a consequence, other appliances can decide to modify their consumption scheduling by changing their strategy according to the new tariffs.
The iterative process is repeated until convergence is reached.
Once convergence is reached, appliances power scheduling and energy prices are fixed as well as the energy bill charged to each householder $h$, which is simply the sum of all his appliances prices $\sum_{t \in {\cal T}} y_{ht} \cdot c_{t}(y_{t})$.

The best response mechanism is executed by solving, in an iterative way, an optimization model.
Specifically, at every iteration and based on the energy demands of other appliances, this model is used to optimally decide the power plan of the appliance in charge of defining its energy demand at this step of the iterative process, with the goal of minimizing the electricity bill.
We will show in the Numerical Results section that our proposed algorithm converges, in few iterations, to a Nash equilibrium.

Note that the best response dynamics here proposed is only used to identify and study the efficiency of the Nash Equilibrium of the game.
While the transmission of the power profile to other users may raise security and privacy concerns, we observe that each appliance needs only the aggregated power profile of other appliances for the real implementation of the Single-Appliance DSM.
Therefore, we can envisage a system in which appliances communicate only with the operator that broadcasts the aggregated information after collecting all appliance's schedules.
Furthermore, to completely remove the communication of any sensitive information, any learning algorithm could be designed to allow appliances to converge to the equilibrium.

\subsection{Multiple-Appliance Game Formulation}
\label{multiAppliance}
%
%
The natural extension of the single-application power scheduling game considers as a player the householder $h$ who chooses the schedule of \textit{all} his appliances according to his preferences.
The strategy space for player $h$ is therefore composed of all variables $x_{aht}$ corresponding to the activities of all his appliances.

\begin{definition}[Multiple-Appliance Power Scheduling Game]
Mathematically, the multiple-appliances power scheduling game is formalized as follows:
\small
\begin{align}
	G:\ & \min_{{\cal I}_h} \ P_h({\cal I}_h, {\cal I}_{-h}) = \sum_{t \in {\cal T}} y_{ht} \cdot c_{t}(y_{t}), \ \forall n\in{\cal N} \\
	 {\cal I}_h = \bigg\lbrace
		 & X_{h} =
		\begin{pmatrix}
		x_{1h1} & x_{1ht} & \cdots & x_{1h \vert {\cal T}\vert} \\
		x_{2h1} & x_{2ht} & \cdots & x_{2h \vert {\cal T}\vert} \\
		\vdots  & \vdots  & \ddots & \vdots  \\
		x_{\vert {\cal A}\vert h1} & x_{\vert {\cal A}\vert ht} & \cdots & x_{\vert {\cal A}\vert h \vert {\cal T}\vert}
		\end{pmatrix} \in \{0,1\}^{\vert {\cal A}\vert \times \vert {\cal T}\vert} : \nonumber\\
		& 	\sum_{t=ST_{ah}}^{ET_{ah} - d_{ah} + 1} x_{aht} = 1 \qquad \forall a \in {\cal A} \label{const_6}\\
		& 	y_{ht} = \sum_{a \in {\cal A}} \sum_{f \in {\cal F}: f \leq t} l_{ahf} x_{ah(t-f+1)} \qquad \forall t \in {\cal T} \label{const_7}\\
		& 	y_{ht} \leq \pi_{SL} \qquad \qquad \qquad \forall  t \in {\cal T} \label{const_8}
		 \bigg\rbrace.
\end{align}
\label{def:ma_game}
\end{definition}
\normalsize
Similarly to~\bkref{const_1},~\bkref{const_4} and~\bkref{const_5}, constraints~\bkref{const_6},~\bkref{const_7} and~\bkref{const_8} are used, respectively, to guarantee that each appliance $a$ starts in exactly one time slot within the interval $(ST_{ah}, ET_{ah})$, to define the daily power demand of house $h$ and to upper-bound the demand of each house according to the supply limit $\pi_{SL}$.

We underline that scheduling optimally multiple appliances increases the complexity of the Smart Grid architecture, since each house requires a central server that collects the energy consumption information from all house appliances and the householder's preferences (i.e., starting/ending times).
Conversely, in the single-appliance formulation each appliance operates independently, and the householder can configure asynchronously the different appliances preferences.
Furthermore, as we will show in the next Section, the higher complexity of the multiple-appliance scheduling game does not result in lower costs for the householder or a lower power peak for the retailer's grid.




\section{Numerical Results}
\label{Numerical_Results}

This section presents the numerical results we obtained evaluating the Single-Appliance DSM (SA-DSM), and the Multiple-Appliance DSM (MA-DSM) mechanisms in realistic Smart Grid scenarios using real traces.
First, we describe the considered scenarios and parameters used in our numerical evaluation.
Then, we compare and discuss the performance achieved by the two proposed mechanisms. 

\subsection{Simulated Scenarios}
\label{caseStudyDescription}
We considered a set $\cal T$ of 24 time slots of 1 hour each.
Residential houses are equipped with 4 shiftable devices out of 11 realistically-modeled appliances\footnote{Namely, \textit{shiftable} devices: washing machine, dishwasher, boiler, vacuum cleaner; \textit{fixed} devices: refrigerator, purifier, lights, microwave oven, oven, TV, iron.}.
Moreover, each house is connected to the grid with a peak power limit of 3~kW ($ \pi_{SL} = 3~\mathrm{kW}$).
The basic domestic configuration and load profiles of each appliance have been defined based on data collected from 100 houses served by an Italian energy supply operator.

Starting from the basic house configuration, we defined multiple scenarios by varying the number of users participating in the game and the parameters of both the energy price function and the scheduling constraints.
Specifically, for the number of houses we considered 3 different cases where the game is played, respectively, by $5$, $20$ and $50$ householders, to assess the performance of the proposed system when the competition level increases. 
Concerning the electricity tariffs, we consider the following pricing function to compute the price paid for the electricity in each time slot $t \in \cal T$:
\small
\begin{align}
	& c_{t}(y_{t}) = 
		\begin{cases}
		c_{MIN} + s \cdot y_{t} \qquad & \forall t \in {\cal T}:y_{t} < \pi_{TT} \\
		c_{MIN} + s \cdot \pi_{TT} \qquad & \forall t \in {\cal T}:y_{t} \geq \pi_{TT}.
		\end{cases}
\label{eq:price_func}
\end{align}
\normalsize
%
In such equations, $y_t$ is the total power demand, $\pi_{TT}$ is a tariff power threshold, $c_{MIN}$ is the minimum electricity price and $s$ is the slope of the cost function. 
Specifically, we fixed the minimum electricity price $c_{MIN}~=~50~\times~10^{-6}$~\euro, and we varied the slope of the cost function by defining it as an integer-multiple of the minimum slope $s_{MIN}~=\dfrac{~0,11~\times~10^{-6}}{\vert{\cal H}\vert}$\euro/kWh, $\vert{\cal H}\vert$ being the number of householders.
As for the value of the tariff threshold, $\pi_{TT}$, after which the energy price is no longer dependent on users' demand, we considered 5 different cases: $25\%$, $30\%$, $35\%$, $40\%$ and $100\%$ of the maximum peak power limit of the whole group of users (i.e., $\vert{\cal H}\vert \cdot \pi_{SL}$). 
By varying the cost function parameters, we assess the impact of the energy tariff on the system performance.

Finally, we also defined different scenarios considering various levels of appliances flexibility.
As reported in Section~\ref{system_model}, for each appliance a bound has been introduced for both the starting and ending time (i.e., $ST_{ah}$ and $ET_{ah}$), representing the period in which the appliance activity has to be executed (note that the activity duration $d_{ah}$ is fixed and lower than the window $ET_{ah}-ST_{ah}$).
Therefore, the larger the execution window is, the higher the system flexibility is in scheduling devices.
In order to evaluate the effect of the scheduling flexibility on the system performance, we defined 3 scenarios where residential users have (1) no flexibility (the \textit{``fix''} label in the following curves), (2) tight flexibility (\textit{``short''}) and (3) loose temporal constraints (\textit{``long''}) on the execution of shiftable devices. 
Moreover, for each of these flexibility levels, we considered two different cases to define $ST_{ah}$ and $ET_{ah}$ for each house: in the first, the parameters of all houses devices are identical (\textit{homogeneous} case), while in the second, we varied them to consider a group of \textit{heterogeneous} users.  

In order to gauge the performance of the proposed mechanisms, we measured the following performance metrics:
\begin{itemize}
  \item \textit{Social Cost}: $P({\cal I})$, defined as in eq. (\ref{eq.SocialCost}). Note that this value represents the electricity bill of the group of houses.
  \item \textit{Fairness}: we considered the \textit{Jain's Fairness Index (JFI)} defined as in~\cite{jain_book}.
 	\item \textit{Peak demand}: defined as the peak of the power demand of the whole group of users: $ \max_{t } \sum_{h \in \mathcal{H}} y_{ht}$. 	
\end{itemize}
%

\subsection{SA-DSM versus MA-DSM}
\label{sec:sa_vs_ma_dsm}
%
\begin{figure*}[t]
	\centering
	\subfloat[SA-DSM Social Cost]
	{
		\includegraphics[width=0.45\textwidth]{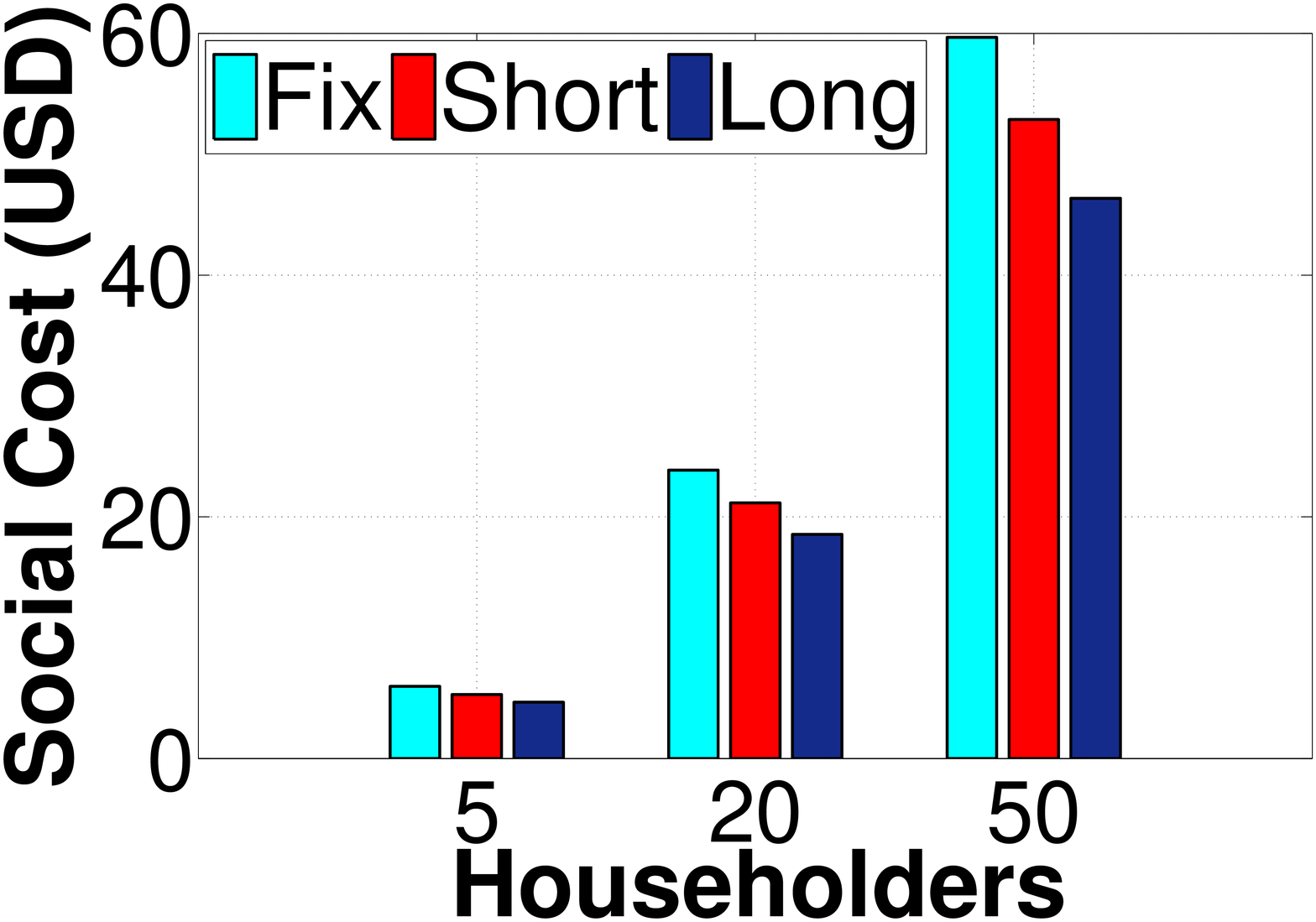}
		\label{sa_sw}
	}
	\subfloat[MA-DSM Social Cost]
	{
		\includegraphics[width=0.45\textwidth]{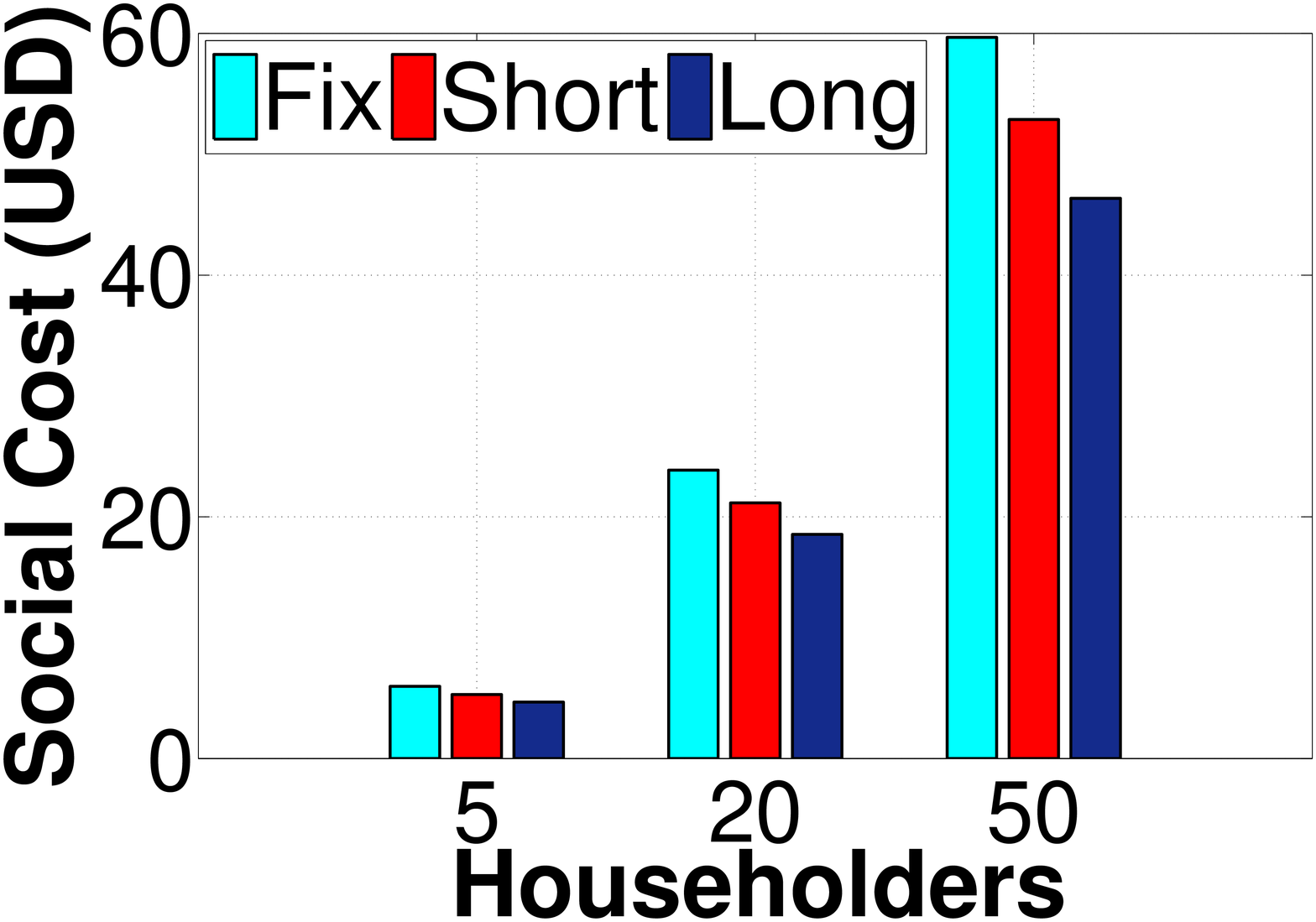}
		\label{ma_sw}
	}\\
	\subfloat[SA-DSM Peak Demand]
	{
		\includegraphics[width=0.45\textwidth]{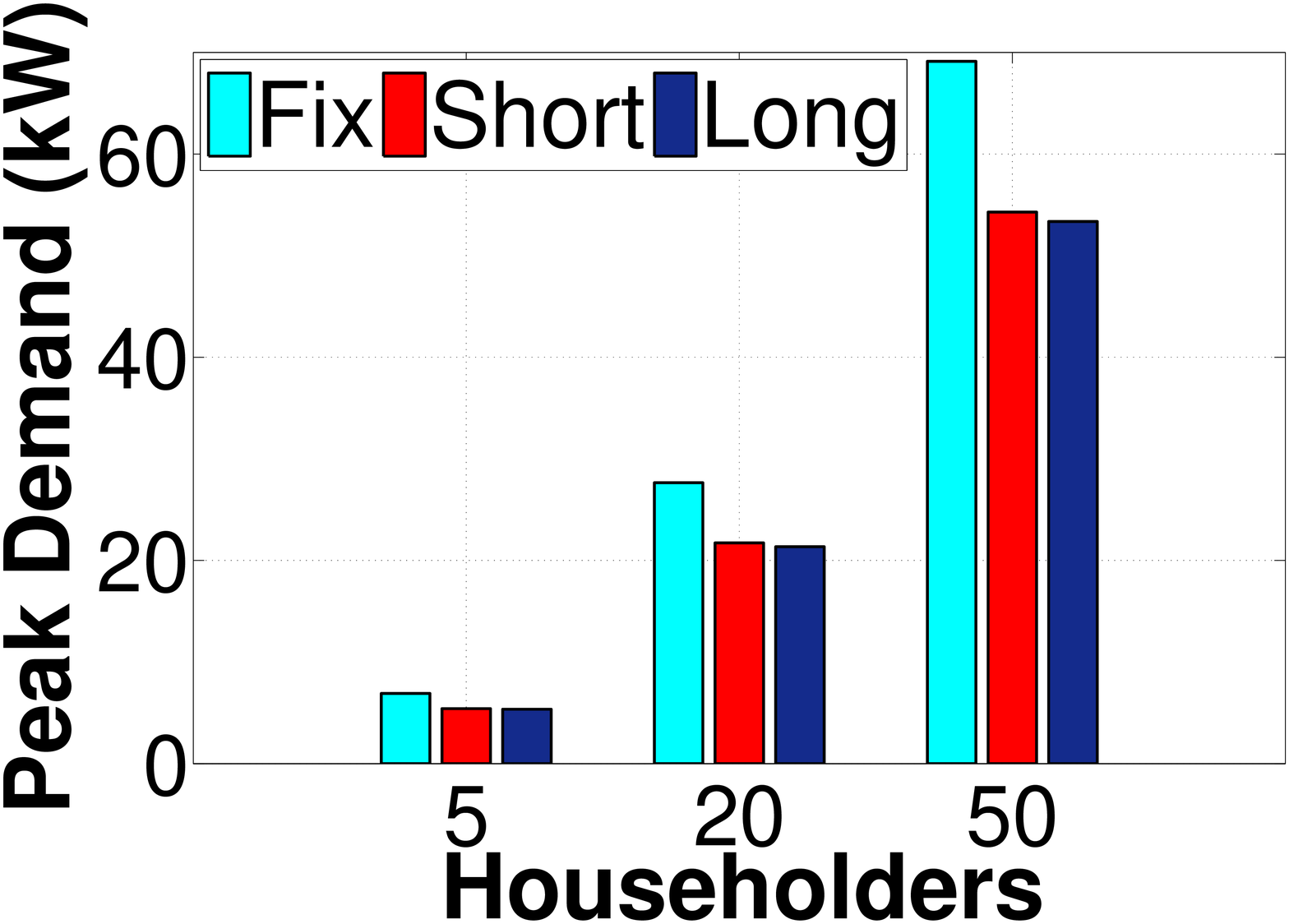}
		\label{sa_peak}
	}
	\subfloat[MA-DSM Peak Demand]
	{
		\includegraphics[width=0.45\textwidth]{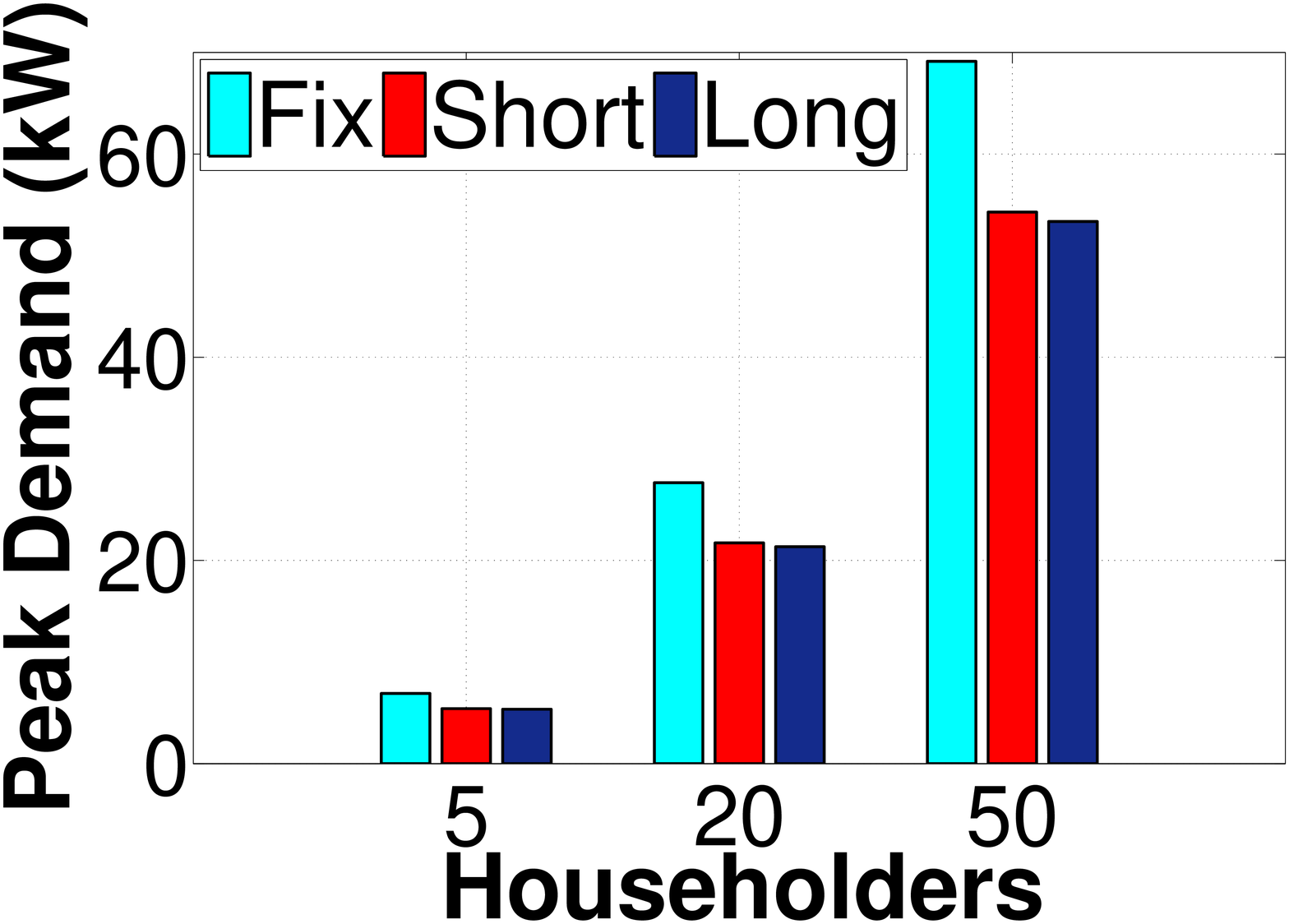}
		\label{ma_peak}
	}
	\caption{\small{SA-DSM versus MA-DSM system results considering homogeneous houses and a linear increasing cost function with minimum slope.}}
	\label{fig:sa_vs_ma}
\end{figure*}
%
Figure~\ref{fig:sa_vs_ma} illustrates the social cost and the peak demand obtained using the two proposed mechanisms as a function of the number of houses.
In such scenario, householders have homogeneous preferences (i.e., $ST_{ah}$ and $ET_{ah}$ vary only among appliances, but all houses' preferences are identical).

It can be observed that both mechanisms exhibit very similar trends in terms of social cost and peak demand.
Indeed, in all the considered scenarios, the gap between the overall householder's electricity bill obtained using the SA-DSM and the MA-DSM is always lower than $3\%$.

The only remarkable difference that we observed between these two solutions is related to the solving time of the corresponding best response dynamics.
Specifically, the SA-DSM mechanism converges more quickly to the Nash Equilibrium than MA-DSM due to the smaller solution space explored by the best response algorithm.
Specifically, in the scenario with 50 houses and long flexibility preferences, the SA-DSM mechanism takes only 8 seconds, in average, to find the equilibrium, whereas the MA-DSM approach needs around 15 minutes\footnote{On an Intel Core i5 3.33~GHz, with a 4~GB RAM.}.
For this reason, the SA-DSM system can be considered an excellent solution for scheduling the appliances execution, since it achieves practically the same results of the MA-DSM system in terms of electricity bills and peak demand, but in a remarkably lower time and with a fully distributed approach.
As a consequence, devices that individually take scheduling decisions represent an effective and efficient solution for realistic Smart Grids deployments: only minimal computation and communication capacity is required among all system's components, without any centralized house controller.

%
It can be further observed from Figures~\ref{fig:sa_vs_ma}\subref{sa_sw} and~\ref{fig:sa_vs_ma}\subref{ma_sw} that, independently of the DSM mechanism, users always benefit from higher scheduling flexibility.
Indeed, larger execution intervals for shiftable appliances (i.e., the curves identified by ``Long'' in the figures) always allow users to pay cheaper bills than those obtained with short and fixed flexibility levels (i.e., curves identified by ``Short'' and ``Fix'', respectively), since the DSM system can explore a larger solution space.
However, the cheaper bills obtained using the long flexibility preferences come at the cost of longer solving time (i.e., the amount of time required to find the Nash Equilibrium through the best response algorithm).
Indeed, we observed that the solving time of the long flexibility scenario doubles with respect to the short flexibility case.
Numerical results presented in Figures~\ref{fig:sa_vs_ma}\subref{sa_sw} and~\ref{fig:sa_vs_ma}\subref{ma_sw} also show that the number of players marginally affects the gain that is achieved with the proposed DSM systems.
In particular, the electricity bill saving obtained with respect to the \textit{no-flexibility} scenario is around 11\% and 22\% for, respectively, the \textit{short-flexibility} and the \textit{long-flexibility scenarios}, irrespective of the number of players and the DSM mechanism.
Indeed, while a larger set of players increases the competition, the proposed DSM mechanisms achieve the same gains by efficiently exploiting the flexibility of shiftable appliances.

One of the main advantages for the operator to adopt the proposed SA-DSM system, as illustrated in Figure~\ref{fig:20_power_profile}, is that it automatically ensures the reduction of the electricity demand during peak hours (i.e., high-price hours) without any centralized coordination among users.
Specifically, the peak demand decreases by as much as 22\% using the SA-DSM system with respect to the value obtained considering fixed scheduling choices (i.e., the \textit{no-flexibility} scenario), and the gain is slightly influenced by the appliances flexibility.
The reduction of the peak power demand results from shifting loads from peak hours to other time-slots.
To this end, only few users' scheduling changes are required (i.e., only appliances used at peak hours have to be shifted) and even a short flexibility can achieve remarkable results.

\begin{figure}[ht]
	\centering
	\includegraphics[width=0.5\textwidth]{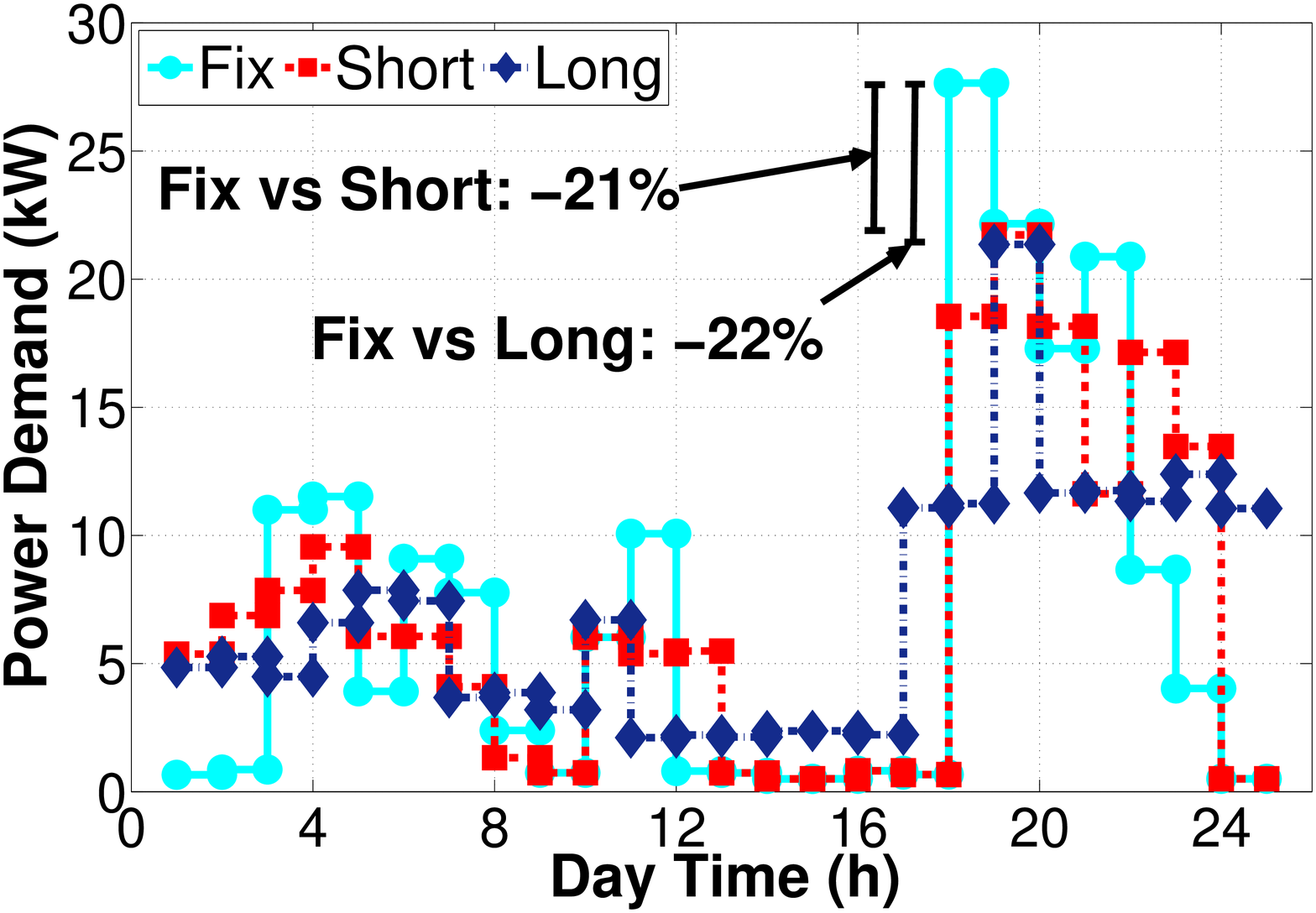}
	\caption{\small{Peak reduction guaranteed by SA-DSM: aggregate power demand of 20 identical houses (80 homogeneous appliances).}}
	\label{fig:20_power_profile}
\end{figure}

\subsection{Analysis of Householder Preferences}
\label{homogeneous_vs_hetherogeneous}
%
\begin{figure*}[t]
	\centering
	\subfloat[Social Cost]
	{
		\includegraphics[width=0.32\textwidth]{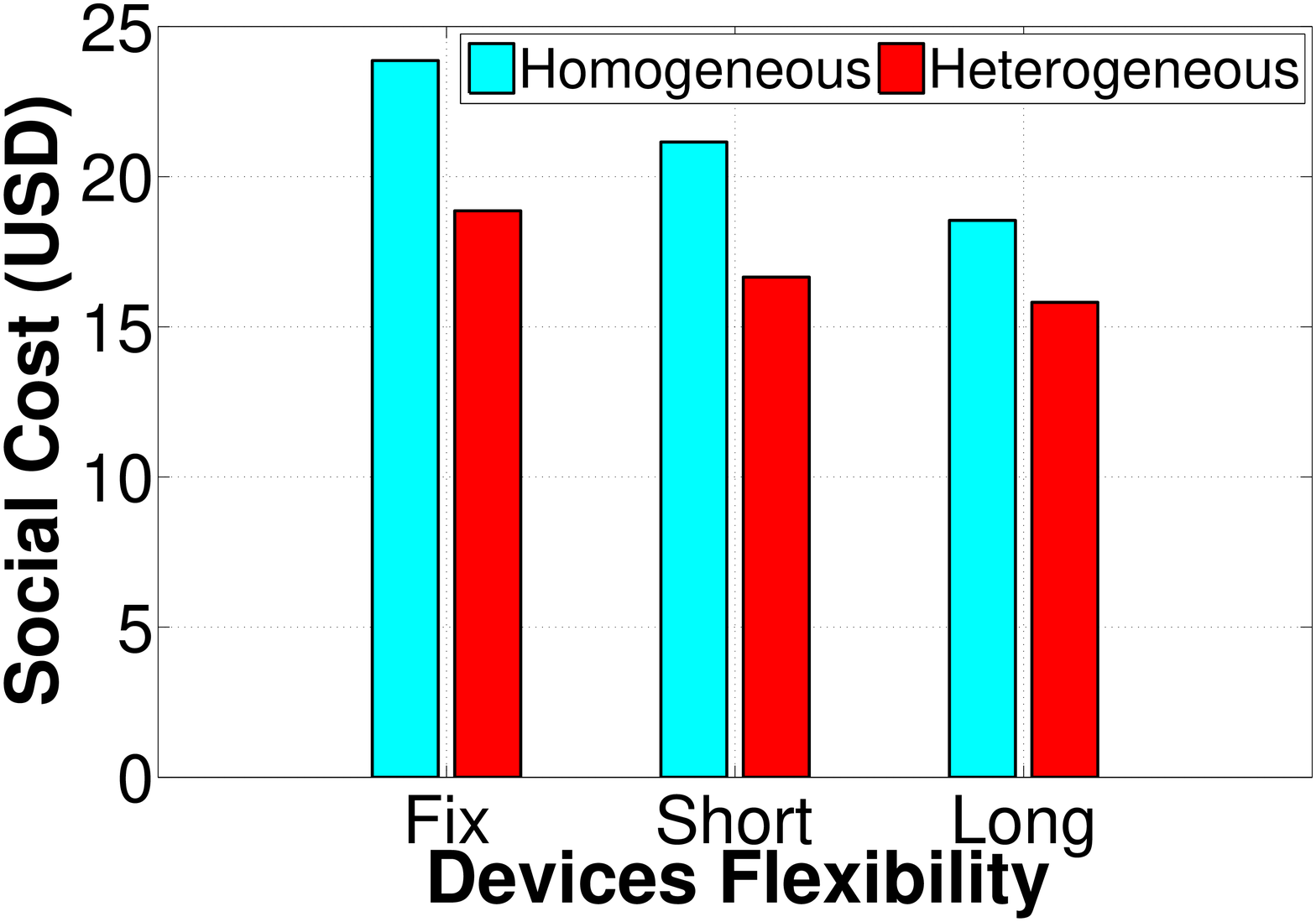}
		\label{hh_sa_sw}
	}
	\subfloat[Peak Demand]
	{
		\includegraphics[width=0.32\textwidth]{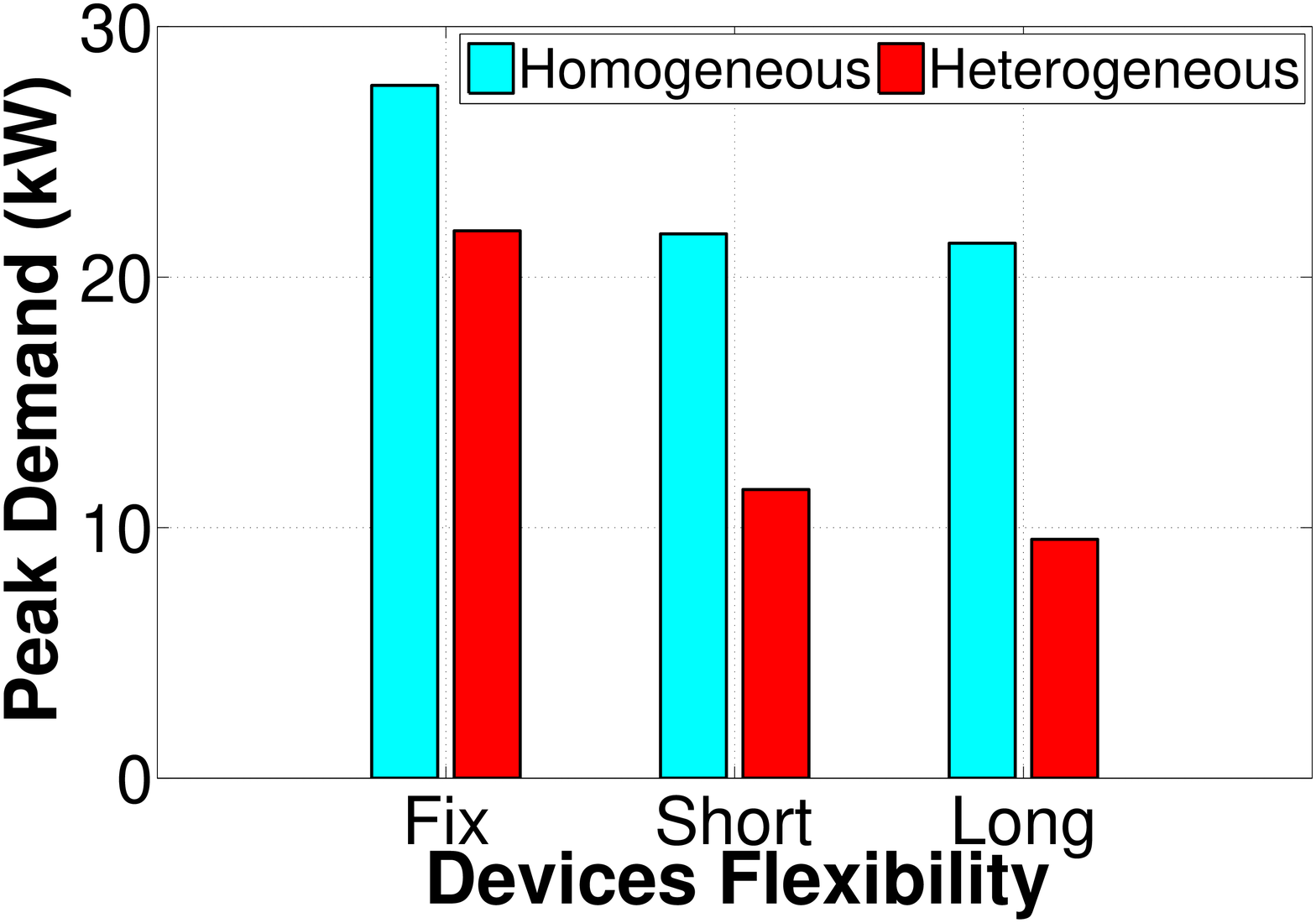}
		\label{hh_sa_peak}
	}
	\subfloat[Power Profile]
	{
		\includegraphics[width=0.32\textwidth]{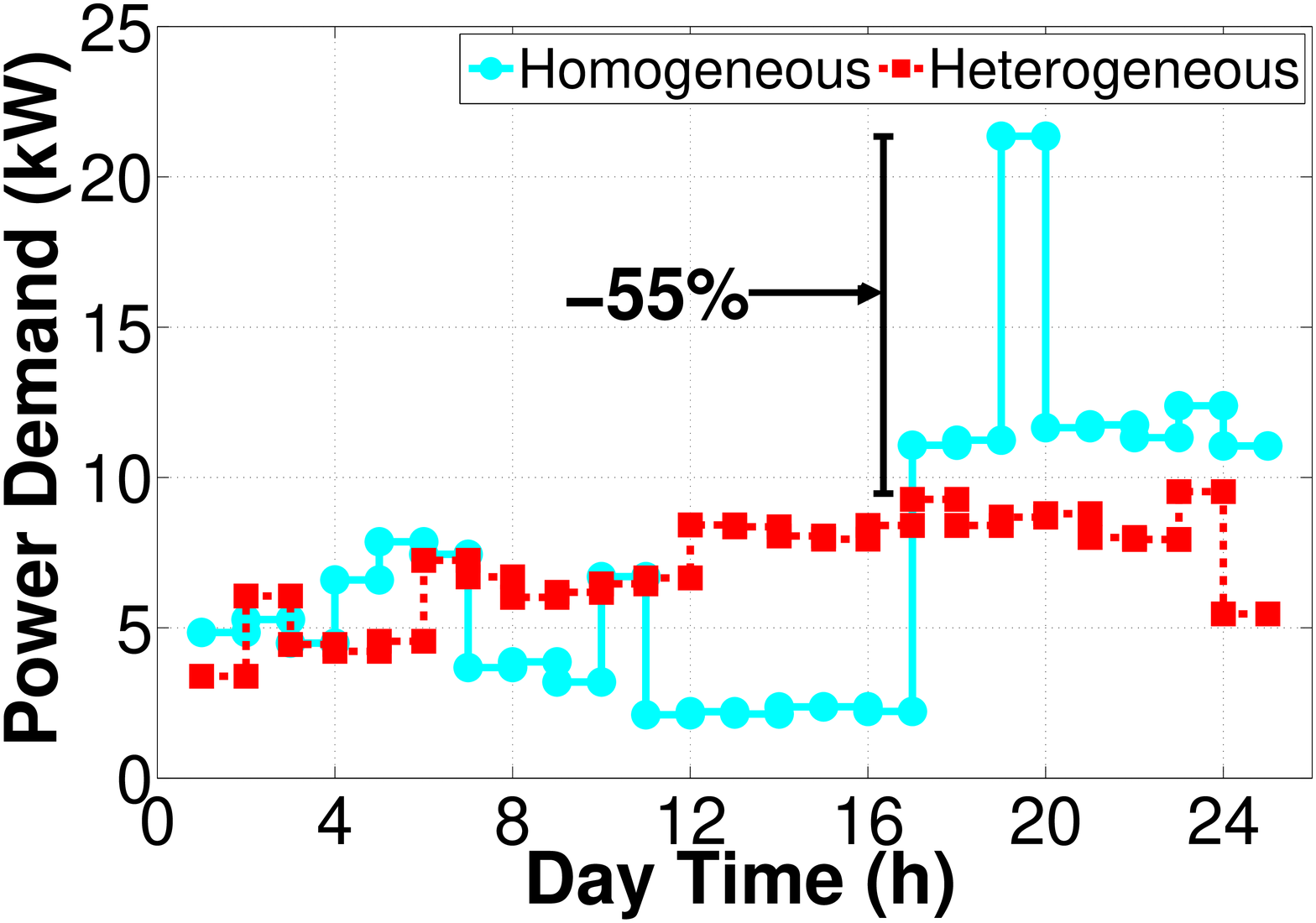}
		\label{hh_sa_power_profile}
	}
	\caption{\small{SA-DSM results with 20 homogeneous and heterogeneous houses preferences.}}
	\label{fig:hom_vs_het}
\end{figure*}
%
Figures~\ref{fig:hom_vs_het}\subref{hh_sa_sw},~\ref{fig:hom_vs_het}\subref{hh_sa_peak} and~\ref{fig:hom_vs_het}\subref{hh_sa_power_profile} illustrate, respectively, the social cost, the peak demand and the aggregated power profile of the proposed SA-DSM mechanism as a function of the appliances flexibility.
Specifically, these figures compare the results obtained with 20 \textit{homogeneous} and \textit{heterogeneous} houses.

As illustrated in Figure~\ref{fig:hom_vs_het}\subref{hh_sa_sw}, the electricity bill is cheaper when considering heterogeneous players.
Indeed, the power demand of heterogeneous houses can be more smoothly distributed over the day than in the homogeneous scenario, due to the different householders preferences about the time windows in which devices can operate.
As a consequence, since the energy price in every time slot is defined as a function of the power demand of houses appliances in that particular slot, players can benefit from loads spreading over time.
Figure~\ref{fig:hom_vs_het}\subref{hh_sa_peak} shows that also the peak demand can be considerably reduced when considering heterogeneous houses.
Specifically, in this case, the proposed SA-DSM mechanism reduces the peak of the power demand down to 55\% in the long flexibility case with respect to the corresponding homogeneous scenario because of a smoother load distribution.
This effect appears clearly in Figure~\ref{fig:hom_vs_het}\subref{hh_sa_power_profile}, where the overall electricity demand over the 24 hours of 20 heterogeneous houses with loose scheduling preferences (long flexibility) is compared to that of 20 identical residential houses.

\subsection{Analysis of Energy Tariffs}
\label{energy_tariffs}
To evaluate how energy tariffs affect the performance of the proposed DSM systems, we fix the slope of the electricity pricing function $s~=\dfrac{~0,11~\times~10^{-6}}{\vert 20 \vert}$\euro/kWh and we consider four different tariff thresholds (i.e., the threshold on the aggregated demand above which the electricity price becomes constant): $\pi_{TT} \in \{15, 18, 21, 24\}~\mathrm{kW}$. 
Figures~\ref{cusers_fig}\subref{bill_cusers} and ~\ref{cusers_fig}\subref{peak_cusers} show the social cost and peak demand of a group of 20 identical houses as a function of the devices flexibility considering the four aforementioned energy tariffs.
As expected, in all cases, the flexibility on the scheduling preferences reduces both the electricity bill and the peak demand.
However, by playing with the energy tariff, the operator can further increase users' gain on the electricity price and, at the same time, decrease the peak power absorbed from the grid, thus resulting in lower investments and operating costs.
For example, as illustrated in Figure~\ref{cusers_fig}\subref{bill_cusers}, the social cost decreases down to 11\% from the no-flexibility to the long flexibility scheduling scenarios when the operator fixes the tariff threshold $\pi_{TT}=15~\mathrm{kW}$.
However, this gain increases up to 22\% with $\pi_{TT}=24~\mathrm{kW}$.
Indeed, when $\pi_{TT}= 15~\mathrm{kW}$, cost savings can be obtained only by shifting loads from peak hours to time slots in which the total power demand is lower than $15~\mathrm{kW}$.
In contrast, a wider set of scheduling alternatives is available to reduce the social cost when $\pi_{TT}= 24~\mathrm{kW}$, since power loads can be shifted from peak hours to all time slots where the aggregated power demand is lower than $24~\mathrm{kW}$.
As a consequence, as the tariff threshold increases, the number of devices shifted outside the peak hours grows, reducing the peak demand as illustrated in Figure~\ref{cusers_fig}\subref{peak_cusers}.

\begin{figure}[t]
	\centering
	\subfloat[Social Cost]
	{
		\includegraphics[width=0.45\textwidth]{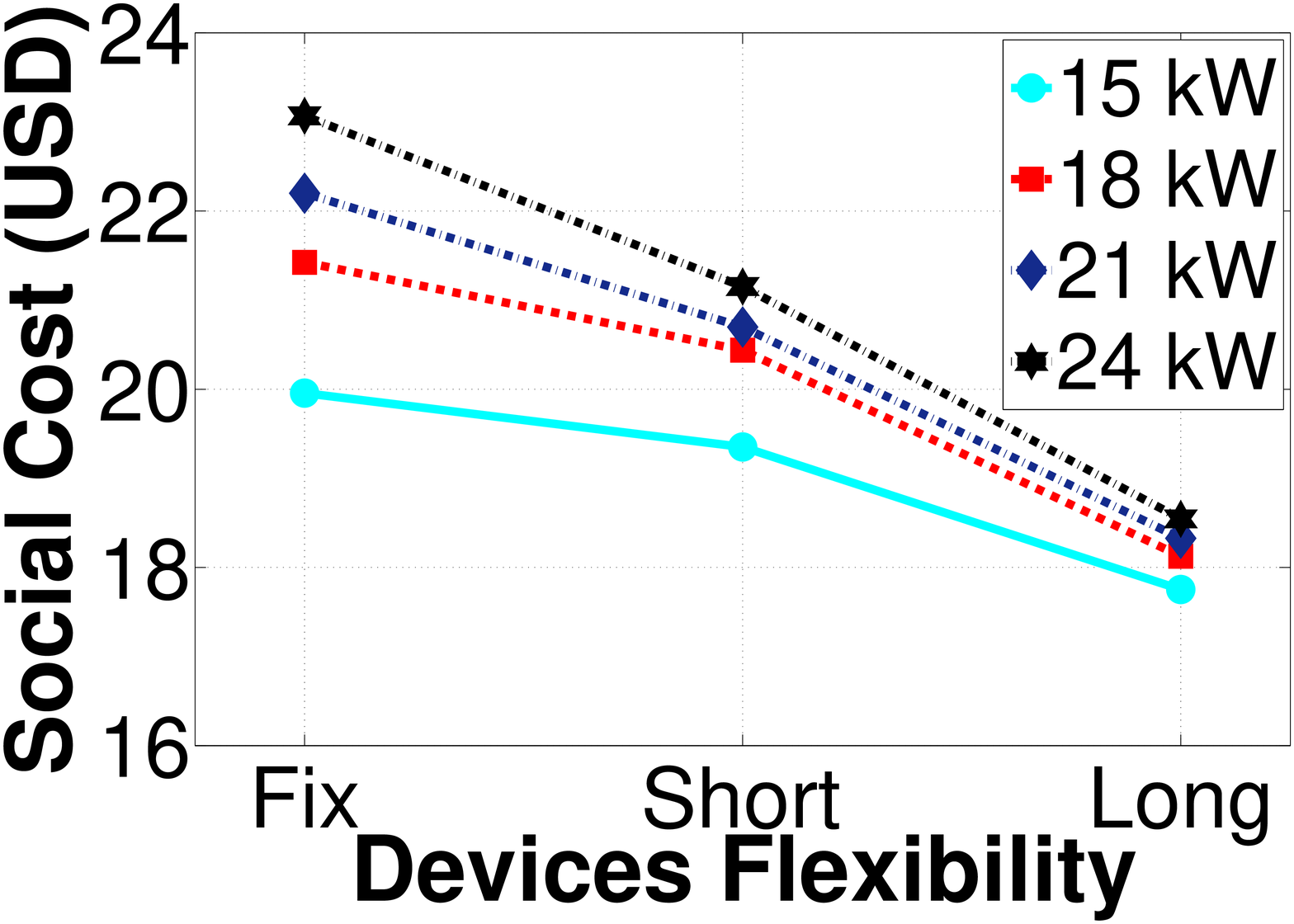}
		\label{bill_cusers}
	}
	\subfloat[Peak Demand]
	{
		\includegraphics[width=0.45\textwidth]{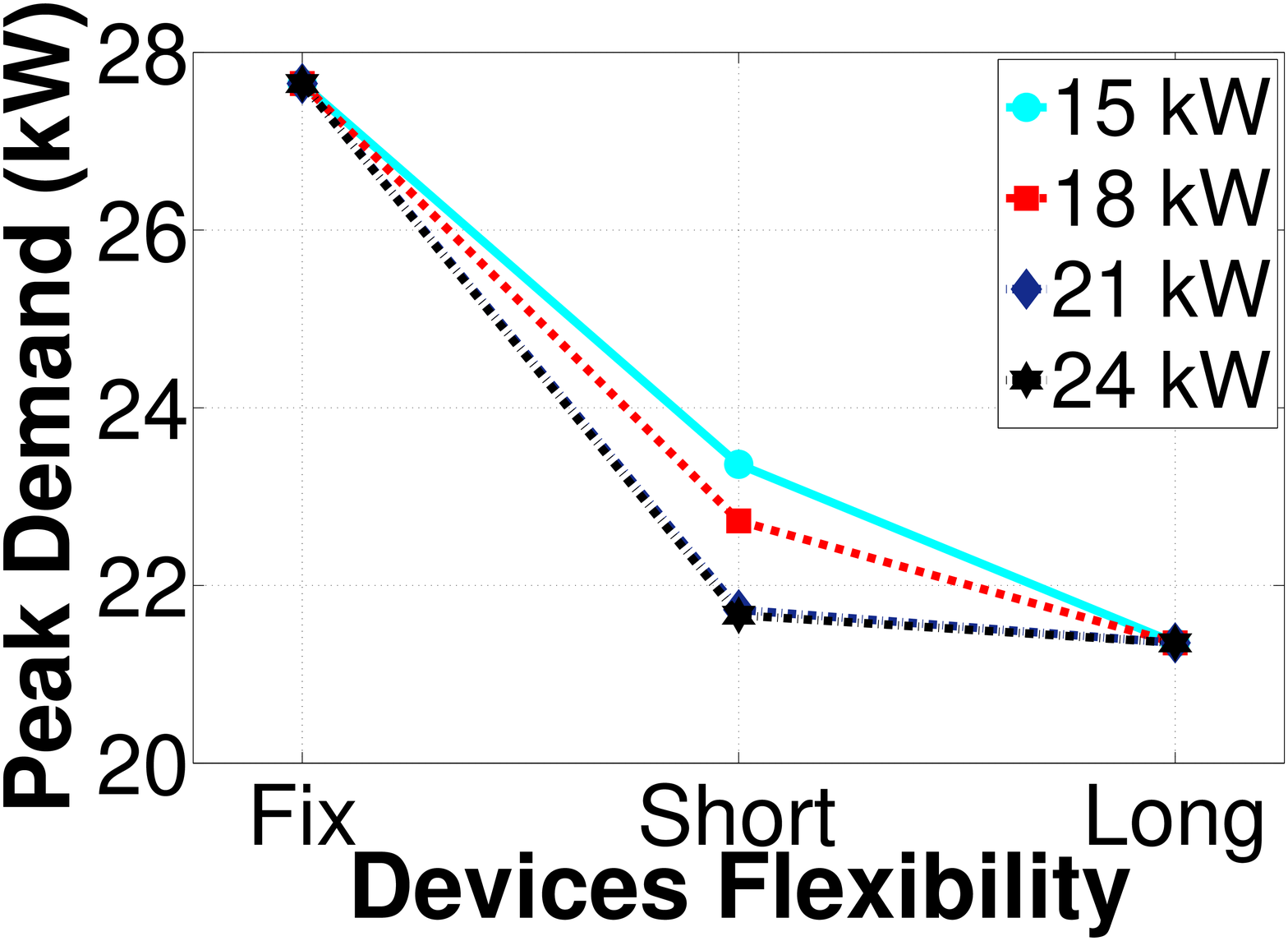}
		\label{peak_cusers}
	}
	\caption{\small{SA-DSM results with 20 identical houses and a varying tariff threshold ($\pi_{TT}$) of the pricing function.}}
	\label{cusers_fig}
\end{figure}
%
In our tests, we also varied the slope $s$ of the energy tariff to assess its impact on the system performance.
However, in the numerical results, which we do not show for the sake of brevity, we observed no significant variation neither on the social cost nor on the peak demand.
Finally, we underline that in all the considered scenarios, we observed that all players pay actually an equal share of the electricity bill, since the Jain's Fairness Index is always very close to 1.
Indeed, even in the scenarios with heterogeneous residential users, the JFI is always higher than $0.9991$.

\section{Conclusions}
\label{conclusion}

In this paper,
we proposed a novel, fully distributed Demand-Side Management (DSM) system aimed at reducing the peak demand of a group of residential users. 

We modeled our system using a game theoretical approach, where players are the customer's appliances, which decide \textit{autonomously} when to execute. 
We demonstrated that the proposed game is a generalized ordinal potential one, 
and we proposed a best response dynamics mechanism which is guaranteed to converge in few steps to efficient Nash equilibrium solutions.
Furthermore, we showed that our approach performs extremely close to a more complex setting where each customer must optimize the schedule of all his appliances, since it provides practically the same results in terms of minimizing their daily electricity bill.
For this reason, due to its intrinsic simplicity, robustness and distributed architecture, we recommend the adoption of our proposed approach.

%
Numerical results, obtained using realistic load profiles and appliance models, demonstrate that the proposed DSM system represents a promising and very effective solution to reduce the peak absorption of the entire system and the electricity bill of individual customers in a fully distributed way. 


\section*{Acknowledgment}
\noindent
This work was partially supported by the French ANR in the framework of the Green-Dyspan project.

\addcontentsline{toc}{section}{\textit{References}}
\bibliographystyle{unsrt}
\bibliography{bibliography}

\clearpage

\end{document}